\documentclass[11pt]{article}
\usepackage[utf8]{inputenc}
\usepackage{style}
\usepackage{framed}
\usepackage[parfill]{parskip}
\usepackage{setspace}
\usepackage{hyperref}

\allowdisplaybreaks

\title{A Note on Approximability of Densest At-Least-$k$-Subgraph}

\author{
Bundit Laekhanukit\thanks{Email: \texttt{bundit.laekhanukit@insait.ai}.}\vspace{-0.5em}\\
INSAIT, Sofia University St. Klement Ohridski, Sofia, Bulgaria.
\and
Pasin Manurangsi\thanks{Email: \texttt{pasin@google.com}.}\vspace{-0.5em}\\
Google Research, Bangkok, Thailand.
\and
Ohad Trabelsi\thanks{Email: \texttt{otrabelsi@cs.haifa.ac.il}}\vspace{-0.5em}\\
University of Haifa, Haifa, Israel.}

\begin{document}

\maketitle

\begin{abstract}
We study the {\sc Densest At-Least-$k$-Subgraph (DAL$k$S)} problem, in which we are given an undirected graph $G$ and an integer $k$, and the goal is to find a subgraph of $G$ with at least $k$ vertices with maximum density. The best-known algorithm, independently discovered by Khuller and Saha~\cite{KS09} and by Andersen~\cite{And07}, yields a 2-approximation for {\sc DAL$k$S} in polynomial time.


In this note, we provide a (simple) reduction from {\sc Densest $k$-Subgraph (D$k$S)} to {\sc Densest At-Least-$k$-Subgraph}, which shows that, if {\sc D$k$S} is hard to approximate to within any constant factor, then {\sc DAL$k$S} is hard to approximate to within $(3/2 - \varepsilon)$ factor for every $\varepsilon > 0$. This holds in both the normal (non-parameterized) and the parameterized (by $k$) settings.


We then generalize the reduction to provide a tight $(2 - \varepsilon)$ factor hardness of approximating {\sc Densest At-Least-$k$-Subgraph}, albeit under a stronger hypothesis which roughly states that {\sc Densest $k$-Subgraph} is hard to approximate to within $k^{1 - \delta}$ factor for any constant $\delta > 0$. Once again, this extends naturally to the parameterized setting. Previously, $(2 - \varepsilon)$ factor inapproximability for DAL$k$S was only known under the Small Set Expansion Hypothesis~\cite{Bergner-thesis,Man17-dalks}, which does not apply to the parameterized version of the problem.


Furthermore, we show that the exact version of DAL$k$S is $\W[1]$-hard (parameterized by $k$).
\end{abstract}

\section{Introduction}

We study the {\sc Densest At-Least-$k$-Subgraph (DAL$k$S)} problem, where we are given an undirected graph $G$ and an integer $k$, and the goal is to find a subgraph of $G$ containing at least $k$ vertices with maximum density. Here the density of a graph is defined as the ratio between the number of edges and the number of vertices. This problem was introduced by Andersen and Chellapilla who presented a 3-approximation algorithm for the problem~\cite{AC09}. Later, an improved 2-approximation algorithm was independently discovered by Khuller and Saha~\cite{KS09} and Andersen~\cite{And07}; the former authors also proved that solving the problem exactly is NP-hard. Interestingly, it was then shown independently by Bergner~\cite{Bergner-thesis} and Manurangsi~\cite{Man17-dalks} that, assuming the Small Set Expansion Hypothesis (SSEH)~\cite{RS10}, the problem is hard to approximate to within a factor of $(2 - \varepsilon)$ for any $\varepsilon > 0$, meaning that the aforementioned algorithm is essentially tight for the problem. While such a strong inapproximability result is known under SSEH, no NP-hardness of approximation of the problem is known, even for a factor of 1.001.

This lack of NP-hardness of approximation for {\sc DAL$k$S} is, in fact, not a coincidence. Specifically, this is also the case for a closely related problem, known as {\sc Densest $k$-Subgraph (D$k$S)}, which is the same as {\sc DAL$k$S}, except that the subgraph is asked to be made of \emph{exactly} $k$ vertices instead of \emph{at least} $k$ vertices. {\sc Densest $k$-Subgraph} has been studied long before the introduction of {\sc Densest At-Least-$k$-Subgraph} (see~\cite{FPK01}) and, despite efforts in establishing hardness of approximation for {\sc D$k$S}~\cite{Feige02,Kho06,AAMMW11,BCVGZ12,BKRW17,Man17,ChuzhoyDGT23}, no NP-hardness of approximation for the problem is known, even for a factor of 1.001. Observe that a $\rho$-approximation algorithm for {\sc D$k$S} immediately implies a $\rho$-approximation algorithm for {\sc DAL$k$S}; putting it differently, any hardness of approximation for {\sc DAL$k$S} immediately transfers to a similar hardness of approximation for {\sc D$k$S}. As a result, establishing NP-hardness of approximating {\sc DAL$k$S} is at least as hard as doing so for {\sc D$k$S}, which, as stated earlier, is a challenging open question.

Besides its approximability, another interesting complexity aspect of {\sc DAL$k$S} that has so far been unexplored is its parameterized complexity. A natural parameter for the problem is $k$, the size requirement. Recall that an algorithm is said to be fixed parameter tractable (FPT)\footnote{See e.g. \cite{DowneyF13} for a more detailed introduction on FPT and parameterized complexity.} if it runs in time $f(k) \cdot n^{O(1)}$ for some function $f: \N \to \N$. Prior to this work, no hardness (or algorithmic) result has been shown in this setting. Specifically, it could be that there exists an FPT algorithm for the problem, or that there exists an FPT algorithm that significantly beats the 2-approximation given by the polynomial time algorithm. The latter premise is particularly plausible, in light of the fact that quite a few problems, which are known to be hard under SSEH (in the non-parameterized setting) via ``similarly looking'' reductions, turn out to admit better approximations in FPT time. Examples of such problems include {\sc Minimum $k$-Vertex Cover}\footnote{In {\sc Minimum $k$-Vertex Cover}, we are given an undirected graph $G$ and an integer $k$, and the goal is to determine the minimum number of edges that can be covered by any set of $k$ vertices.} and {\sc Minimum $k$-Cut}\footnote{In {\sc Minimum $k$-Cut}, we are given an undirected graph $G$ and an integer $k$, and the goal is to determine the minimum number of edges to be removed so that the remaining graph consists of at least $k$ connected components.}; similar to {\sc DAL$k$S}, both are SSE-hard to approximate to within factor of $(2 - \varepsilon)$ for any $\varepsilon > 0$~\cite{GandhiK15,Man17-dalks} but both admit FPT approximation schemes~\cite{Marx08,GLL18,Manurangsi19,LokshtanovSS24}. We refer interested readers to the survey~\cite{FeldmannSLM20}, which discusses more results on FPT approximation in detail.

\subsection{Our Results}

Our first result is a ``reverse reduction'' from {\sc D$k$S} to {\sc DAL$k$S}, which shows that, if there is no polynomial time algorithm that approximates {\sc D$k$S} to within any constant factor, then there is no polynomial time algorithm that approximates {\sc DAL$k$S} to within $(3/2 - \varepsilon)$ factor for any $\varepsilon > 0$. The reduction also works in the FPT regime; that is, if we assume further that there is no FPT (parameterized by $k$) algorithm that approximates {\sc D$k$S} to within any constant factor, then there is no FPT algorithm that approximates {\sc DAL$k$S} to within $(3/2 - \varepsilon)$ factor for any $\varepsilon > 0$.

To be more formal, let us introduce a few notations. For any undirected graph $G$, we use $V_G$ and $E_G$ to denote its vertex set and edge set, respectively. For $S \subseteq V_G$, we use $G[S]$ to denote the induced subgraph of $G$ on $S$. Let $\edge_{= k}(G) := \max_{S \subseteq V_G, |S| = k} |E_{G[S]}|$, i.e., the maximum number of edges in any $k$-vertex subgraphs of $G$; $\edge_{\leqs k}(G), \edge_{\geqs k}(G)$ are defined similarly. We use $\den(G)$ to denote the density of $G$, i.e., $\den(G) := |E_G|/|V_G|$. For convenience, we use $\den_{=k}(G)$ to denote $\max_{S \subseteq V_G, |S| = k} \den(G[S])$; $\den_{\leqs k}(G)$ and $\den_{\geqs k}(G)$ are defined in similar manners.

The reductions in the paper will be formalized through gap problems. The starting point of our first result is the gap version of the {\sc Densest $k$-Subgraph} problem, as defined below.

\begin{definition}
For $\lambda \geqs 1$, the {\sc GapD$k$S}$(\lambda)$ problem is: given a graph $G$ and $k, \ell \in \N$, distinguish between
 (Completeness) $\edge_{=k}(G) \geqs \ell$, and,
 (Soundness) $\edge_{=k}(G) < \ell/\lambda$.
\end{definition}

Our first result can then be stated as follows.

\begin{theorem} \label{thm:dks}
Suppose that, for any $\lambda \geqs 1$, no polynomial time algorithm solves {\sc GapD$k$S}$(\lambda)$. Then, for any $\varepsilon > 0$, no polynomial time algorithm approximates {\sc DAL$k$S} to within a factor of $(3/2 - \varepsilon)$.

Moreover, if we assume further that, for any $\lambda \geqs 1$ and $f: \N \to \N$, no $f(k) \cdot n^{O(1)}$ time algorithm solves {\sc GapD$k$S}$(\lambda)$, then, for any $\varepsilon > 0$, no $f(k) \cdot n^{O(1)}$ time algorithm approximates {\sc DAL$k$S} to within a factor of $(3/2 - \varepsilon)$.
\end{theorem}

Although, as stated earlier, no NP-hardness of approximation of {\sc D$k$S} is known, hardness of approximation under stronger assumptions is known. In particular, {\sc D$k$S} is hard to approximate to within any constant factor assuming the Random 3-SAT (R3SAT) Hypothesis\footnote{The Random 3-SAT Hypothesis~\cite{Feige02} states that no polynomial time algorithm can refute random 3SAT formulae correctly with probability 0.5 without falsely refuting any satisfiable formula.}, the Planted Clique Hypothesis\footnote{The Planted Clique Hypothesis~\cite{Jer92,Kuc95} states that no polynomial time algorithm can distinguish between a graph drawn from $\cG(n, 1/2)$ and one in which a clique of size polynomial in $n$ is additionally planted.}~\cite{AAMMW11}, or the Exponential time Hypothesis (ETH)\footnote{The Exponential Time Hypothesis (ETH)~\cite{IP01,IPZ01} states that no $2^{o(n)}$ time algorithm can solve 3SAT.}~\cite{Man17}; in the parameterized setting, {\sc D$k$S} is known to be hard to approximate to within any constant factor assuming Gap-ETH\footnote{The Gap Exponential Time Hypothesis (Gap-ETH)~\cite{Din16,MR17} states that no $2^{o(n)}$ time algorithm can distinguish between a satisfiable 3-SAT formula and one which is not even $(1 - \gamma)$-satisfiable for some $\gamma > 0$.}~\cite{CCKLMT17}. As a result, the following is an immediate consequence of Theorem~\ref{thm:dks}:

\begin{cor}
Assuming Planted Clique Hypothesis, R3SAT Hypothesis, or ETH, there is no polynomial time $(3/2 - \varepsilon)$-approximation algorithm for {\sc DAL$k$S} for any $\varepsilon > 0$. Moreover, assuming Gap-ETH, for any function $f: \N \to \N$, there is no $f(k) \cdot n^{O(1)}$ time $(3/2 - \varepsilon)$-approximation algorithm for {\sc DAL$k$S} for any $\varepsilon > 0$.
\end{cor}

Note that the above result is incomparable to the known SSEH-based inapproximability results~\cite{Bergner-thesis,Man17-dalks}, as the assumptions above are not known to follow from SSEH and vice versa.

While the above hardness provides evidence that there is unlikely to be an FPT approximation scheme for the problem, the obvious remaining question is whether it is possible to (significantly) beat the 2-approximation in FPT time. Although we do not know how to rule out such hardness under the previous assumption that {\sc D$k$S} is hard to approximate to within any constant factor, we can show such a result under a stronger assumption which is stated formally below.

\begin{definition}[PolyGapDkS$(\delta, t)$]
For any $\delta \in (0, 1]$ and $t \in \N$, the {\sc PolyGapDkS}$(\delta, t)$ problem is the following: given a graph $G$ and $k \in \N$, distinguish between 
\begin{itemize}
\item (Completeness) there exists $S \subseteq V_G$ of size $k$ s.t. $G[S]$ contains at least $k^{t - \delta}$ copies of $t$-cliques, and,
\item (Soundness) $\den_{\leqs k}(G) < k^\delta$.
\end{itemize}
\end{definition}

\begin{theorem} \label{thm:polydks}
Suppose that, for any $\delta > 0$ and any $t \in \N$, no polynomial time algorithm solves {\sc PolyGapD$k$S}$(\delta, t)$. Then, for any $\varepsilon > 0$, no polynomial time algorithm approximates {\sc DAL$k$S} to within a factor of $(2 - \varepsilon)$.

Moreover, if we assume further that, for any $\delta > 0, t \in \N$ and $f: \N \to \N$, no $f(k) \cdot n^{O(1)}$ time algorithm solves {\sc PolyGapD$k$S}$(\delta, t)$, then, for any $\varepsilon > 0$, no $f(k) \cdot n^{O(1)}$ time algorithm approximates {\sc DAL$k$S} to within a factor of $(2 - \varepsilon)$.
\end{theorem}

We remark that the completeness in {\sc PolyGapD$k$S}$(\delta, t)$ is indeed stronger than just requiring that $G[S]$ is dense. In particular, there are dense graphs that contain few copies of $t$-clique; for example, a complete bipartite graph does not contain $t$-clique for any $t \ge 3$. Similarly, the soundness is also stronger since it requires the density bound even for subgraphs of size less than $k$.
Nevertheless, it turns out that some hardness results in the literature do satisfy these stronger conditions. In particular, Manurangsi et al.~\cite{ManurangsiRS21} proved, under the Strongish Planted Clique Hypothesis\footnote{The Strongish Planted Clique Hypothesis~\cite{ManurangsiRS21} states that no $n^{o(\log n)}$ time algorithm can distinguish between a graph drawn from $\cG(n, 1/2)$ and one in which a clique of size polynomial in $n$ is additionally planted.}, a $o(k)$-hardness for {\sc D$k$S} that holds even with ``perfect completeness'' where there exists $S \subseteq V_G$ such that $G[S]$ is a $k$-clique. This obviously implies that $G[S]$ contains $k^{t-o(1)}$ copies of $t$-clique. Their soundness also holds for all subgraphs of size at most $k$ (rather than just equal to $k$). Moreover, their hardness result applies even against FPT algorithms. Thus, Theorem~\ref{thm:polydks} yields the following corollary.

\begin{cor}
Assuming Strongish Planted Clique Hypothesis, for any function $f: \N \to \N$, there is no $f(k) \cdot n^{O(1)}$ time $(2 - \varepsilon)$-approximation algorithm for {\sc DAL$k$S} for any $\varepsilon > 0$.
\end{cor}

Finally, we also adapt our reduction slightly to show that {\sc Densest At-Least-$k$-Subgraph} parameterized by $k$ is \W[1]-hard to solve exactly.

\begin{theorem} \label{thm:w1}
{\sc Densest At-Least-$k$-Subgraph} is \W[1]-hard when parameterized by $k$.
\end{theorem}

We remark here that {\sc Densest At-Least-$k$-Subgraph} is in XP\footnote{XP is the class of parameterized problems that can be solved in time $n^{f(k)}$ for some function $f: \N \to \N$, where $n$ is the input size and $k$ is the parameter.}. Note that it is not a priori clear that this is the case; unlike in {\sc Densest $k$-Subgraph}, the solution(s) to {\sc Densest At-Least-$k$-Subgraph} can have large size and hence we may need to enumerate all $2^n$ subsets of vertices to find the optimal solution. Fortunately, Saha \etal~\cite{SHKRZ10} showed that, given a subset $S \subseteq V_G$, it is possible to efficiently find the densest subgraph containing $S$. As a result, we can simply enumerate all subsets $S$'s of size $k$, use Saha \etal's algorithm and output the densest subgraph among these subgraphs; this yields an $n^{k + O(1)}$ time algorithm for {\sc DAL$k$S}.

\section{Preliminaries}

For every $n \in \N$, we use $[n]$ to denote $\{1, \dots, n\}$. 
For any undirected graph $G$, we use $\deg_G(v)$ to denote the degree of $v$ (with respect to $G$) for every vertex $v \in V_G$. 
In the context of parameterized complexity, we consistently use ``$k$'' as our parameter in all problems.
We recall that a reduction is FPT if the parameter of the output instance is bounded by a function of the input parameter, independent of the total instance size.

\subsection{Massaging the Hypotheses}

It will be more convenient for us to use hypotheses that look seemingly stronger than those originally stated ({\sc GapD$k$S}$(\lambda)$ and {\sc PolyGapDkS}$(\delta, t)$) but are actually equivalent to the original hypotheses. We start with a variant of {\sc GapD$k$S} which now takes in two parameters $\lambda, \gamma$ and, in the soundness case, requires a (seemingly) stronger condition that any subgraph of size $\gamma \cdot k$ still contains few edges, as stated below.

\begin{definition}
For any $\lambda, \gamma \geqs 1$, the {\sc GapDkS}$(\lambda, \gamma)$ problem is: given a graph $G$ and $k, \ell \in \N$, distinguish between (Completeness) $\edge_{\leqs k}(G) \geqs \ell$ and (Soundness) $\edge_{\leqs \gamma \cdot k}(G) < \ell/\lambda$.
\end{definition}

It turns out that this additional requirement in the soundness case does not change the hypothesis, as there is a reduction from the original to this variant:

\begin{prop} \label{prop:size-relax}
For any $\lambda, \gamma \geqs 1$, there is an FPT polynomial time reduction from {\sc GapDkS}$(2\lambda\gamma^2)$ to {\sc GapDkS}$(\lambda, \gamma)$.
\end{prop}

\begin{proof}
The reduction is the trivial reduction; that is, given an input $(G, k, \ell)$ for {\sc GapDkS}$(2\lambda\gamma^2)$, it simply outputs $(G, k, \ell)$. It is obvious that this is an FPT reduction. Moreover, since we retain the same graph, the completeness follows trivially. To see the soundness of the reduction, suppose contrapositively that
$\edge_{\leqs \gamma \cdot k}(G) \geqs \ell/\lambda$. In other words, there
exists $S \subseteq V_G$ of size $r \leqs \gamma k$ such that
$|E_{G[S]}| \geqs \ell/\lambda$. If $r < k$, then we may add arbitrary vertices
to $S$ until it has size $k$, which only increases the number of induced edges;
hence, $\edge_{=k}(G) \geqs \ell/\lambda \geqs \ell/(2\lambda\gamma^2)$.
Otherwise, $k \leqs r \leqs \gamma k$. Let $S'$ be a random $k$-element subset
of $S$; the probability that each edge in $E_{G[S]}$ remains in $E[S']$ is
\[
\frac{k(k - 1)}{r(r - 1)}
\geqs
\frac{k(k - 1)}{(\gamma k)(\gamma k - 1)}
\geqs
\frac{1}{2\gamma^2}.
\]
Thus, there exists $S^* \subseteq S$ of size $k$ with
$|E[S^*]| \geqs \frac{\ell}{2\lambda\gamma^2}$. This means that
$\edge_{=k}(G) \geqs \ell/(2\lambda\gamma^2)$, which completes the soundness proof.
\end{proof}

We state a similar variant of {\sc PolyGapDkS} and provide an analogous reduction below.

\begin{definition}[PolyGapDkS$(\delta, t, \gamma)$]
For any $\delta \in (0, 1]$, $\gamma \geqs 1$ and $t \in \N$, the {\sc PolyGapDkS}$(\delta, t, \gamma)$ problem is the following: given a graph $G$ and $k \in \N$, distinguish between 
\begin{itemize}
\item (Completeness) There exists $S \subseteq V_G$ of size $k$ s.t. $G[S]$ contains at least $k^{t - \delta}$ copies of $t$-cliques, and,
\item (Soundness) $\den_{\leqs \gamma \cdot k^t}(G) < k^\delta$.
\end{itemize}
\end{definition}

\begin{prop} \label{prop:size-relax-poly}
For any $\delta \in (0, 1]$, $\gamma \geqs 1$ and $t \in \N$, there exists an FPT polynomial time reduction from {\sc PolyGapDkS}$(0.5\delta/t, t)$ to {\sc PolyGapDkS}$(\delta, t, \gamma)$.
\end{prop}

\begin{proof}
Given an instance $(G, k)$ of {\sc PolyGapDkS}$(0.5\delta/t, t)$, we create an instance $(G, k')$ as follows: the graph $G$ remains unchanged whereas we set $k'$ to be $\lfloor (k/\gamma)^{1/t} \rfloor$. Note that we may assume that $k$ is sufficiently large, as otherwise the problem can be solved in $n^{O(k)}$, which is polynomial when $k$ is constant. Specifically, we assume that $k \geqs 2\gamma \cdot (2^{t + 1} \gamma)^{2t/\delta}$; this also implies that $k' \geqs (2^{t + 1} \gamma)^{2/\delta}$. Observe that this implies that
\begin{align} \label{eq:tmp-bound}
k^{0.5\delta/t}
\leqs (\gamma(2k')^t)^{0.5\delta/t}
\leqs 2^{-t} \cdot (k')^{\delta} \cdot \left(\frac{(k')^{0.5\delta}}{2^{t + 1}\gamma}\right)^{-1}
\leqs 2^{-t} \cdot (k')^{\delta}
\end{align}

It is obvious that this is an FPT reduction; we next proceed to argue its completeness and soundness.

{\bf (Completeness)} Suppose that there exists $S \subseteq V_G$ containing $k$ vertices such that $G[S]$ contains at least $k^{t - 0.5\delta / t}$ copies of $t$-clique. Let $S'$ be a random subset of size $k'$ of $S$. The probability that each $t$-clique in $S$ remains in $S'$ is
\begin{align*}
\frac{\binom{k - t}{k' - t}}{\binom{k}{k'}} = \frac{k'(k' - 1) \cdots (k' - t + 1)}{k(k - 1) \cdots (k - t + 1)} \geqs \left(\frac{k' - t}{k}\right)^t.
\end{align*}
As a result, there must exist $S^* \subseteq S$ of size $k'$ that contains at least
\begin{align*}
\left(\frac{k' - t}{k}\right)^t \cdot k^{t - 0.5\delta/t} &= (k' - t)^t \cdot k^{-0.5\delta/t} \stackrel{\eqref{eq:tmp-bound}}{\geqs} (k'/2)^t \cdot (2^{t} \cdot (k')^{-\delta}) = (k')^{t - \delta}
\end{align*}
copies of $t$-cliques, which concludes the completeness part of the proof.

{\bf (Soundness)} The soundness guarantee of the original instance is that $\den_{\leqs \gamma \cdot (k')^t}(G) < k^{0.5\delta/t}$, which is at most $(k')^{\delta}$ due to~\eqref{eq:tmp-bound}.
\end{proof}

\section{A Strong Hardness for {\sc Densest {\it k}-Subhypergraph}}

As an intermediate problem, we will be reducing to a variant of the {\sc Densest $k$-Subhypergraph (D$k$SH)}. The standard version of {\sc D$k$SH} on $t$-uniform hypergraphs can be stated as follows: Given a universe $U$ and a collection $\cS$ of $t$-element subsets of $U$, find $T \subseteq U$ of size $k$ that entirely contains as many subsets $S \in \cS$ as possible. The variant we work with -- which is formalized below -- is a stronger version, where the completeness remains the same but, in the soundness, only a few $S \in \cS$ even satisfy $|S \cap T| > 1$. (In contrast, the standard gap version would only consider $|S \cap T| = t$.)

\begin{definition}[$t$-uniform StrongGapD$k$SH$(\lambda, \gamma)$]
For any $\lambda, \gamma \geqs 1$ and $t \in \N$, the $t$-uniform {\sc StrongGapD$k$SH$(\lambda, \gamma)$} problem is: given a universe $U$, a collection $\cS$ of $t$-element subsets of $U$ and $k, \ell \in \N$, distinguish between
\begin{itemize}
\item (Completeness) There exists $T \subseteq U$ of size $k$ such that $|\{S \in \cS \mid S \subseteq T\}| \geqs \ell$.
\item (Soundness) For any $T \subseteq U$ of size at most $\gamma \cdot k$, $|\{S \in \cS \mid |S \cap T| > 1\}| < \ell/\lambda$
\end{itemize}
\end{definition}

It is simple to see that {\sc GapDkS}$(\lambda, \gamma)$ can be cast as 2-uniform {\sc StrongGapD$k$SH$(\lambda, \gamma)$}.

\begin{lemma} \label{lem:dks-to-dksh}
For any $\lambda, \gamma \geqs 1$, there exists an FPT polynomial time reduction from {\sc GapDkS}$(\lambda, \gamma)$ to 2-uniform {\sc StrongGapD$k$SH$(\lambda, \gamma)$}.
\end{lemma}

\begin{proof}
Given an instance $(G, k, \ell)$ of {\sc GapDkS}$(\lambda, \gamma)$, we construct an instance $(U, \cS, k, \ell)$ for 2-uniform {\sc StrongGapD$k$SH$(\lambda, \gamma)$} by letting $U = V_G$ and letting each subset $S \in \cS$ correspond to an edge in $E_G$; that is, for each $e = \{u, v\} \in E_G$, create a subset $S_e = \{u, v\}$ and add it into $\cS$. The parameters $k$ and $\ell$ remain unchanged. It is obvious that the reduction is FPT, and the completeness and soundness of the reduction are trivial from definitions.
\end{proof}

Perhaps more interestingly, we show that {\sc PolyGapDkS}$(1/3, t, \gamma)$ can also be reduced to $t$-uniform {\sc StrongGapD$k$SH$(\lambda, \gamma)$}. We remark that the constant $1/3$ here can be changed to arbitrary number strictly less than 1/2; we simply set it to 1/3 to avoid introducing yet another parameter.

\begin{lemma} \label{lem:polydks-to-dksh}
For any $\lambda, \gamma \geqs 1$ and $t \in \N$, there exists an FPT polynomial time reduction from {\sc PolyGapDkS}$(1/3, t, \gamma)$ to $t$-uniform {\sc StrongGapD$k$SH$(\lambda, \gamma)$}.
\end{lemma}

\begin{proof}
Given an instance $(G, k)$ of {\sc PolyGapDkS}$(1/3, t, \gamma)$, we construct an instance $(U, \cS, k', \ell)$ of $t$-uniform {\sc StrongGapD$k$SH$(\lambda, \gamma)$} as follows. The ground set $U$ consists of all the $(t - 1)$-cliques in graph $G$, i.e., $U = \{C_{t-1} \in \binom{V}{t - 1} \mid C_{t-1} \text{ forms a clique in } G\}$. We then create one subset for each $t$-clique in $G$. Specifically, for each $C_t \in \binom{V}{t}$ that forms a clique in $G$, we create a subset $S_{C_t}$ that contains all $(t - 1)$-size subset of $C_t$, i.e., $S_{C_t} = \binom{C_t}{t - 1}$. Finally, let $k' = \binom{k}{t - 1}$ and $\ell = k^{t - 1/3}$.

It is obvious that the reduction is an FPT reduction. We next proceed to prove the completeness and soundness of the reduction. Throughout the proof, we will assume that $k$ is sufficiently large; specifically, that $k \geqs (\lambda\gamma t)^3$. We remark that this assumption can be made without loss of generality since {\sc PolyGapDkS} can be solved in time $|G|^{O(k)}$ which is polynomial for constant $k$.

{\bf (Completeness)} Suppose that there exists a $k$-vertex subgraph $H$ of $G$ that contains at least $k^{t - 1/3} = \ell$ copies of $t$-clique. Let $T \subseteq U$ be the set of all $(t - 1)$-cliques inside $H$. Clearly, $T$ has size at most $\binom{k}{t-1} = k'$. Moreover, $S_C \subseteq T$ for every $t$-clique $C$ in $H$, and there are at least $\ell$ such $C$'s.

{\bf (Soundness)} Assume that $\den_{\leqs \gamma \cdot k^t}(G) < k^{1/3}$. Let us consider any subset $T \subseteq U$ of size at most $\gamma \cdot k'$. Let $Y$ denote all the $t$-cliques $C$'s such that $|S_C \cap T| > 1$. We bound the size of $Y$ by double counting the number of pairs of $(t - 1)$-cliques in $T$ whose union belongs to $Y$. That is, we will count the set $Z = \{(C_1, C_2) \in T \times T \mid C_1 \cup C_2 \in Y\}$ in two ways. First, observe that, for every $C \in Y$ , if we pick any distinct $C_1, C_2$ from $S_C \cap T$, then we have $C_1 \cup C_2 = C \in Y$. As a result, we have $|Z| \geqs |Y|$. In other words, it suffices for us to upper bound the size of $Z$.

To bound $Z$, observe that, for every $(C_1, C_2) \in Z$, we must have $|C_1 \cap C_2| = t - 2$. For each $\Cint \in \binom{V}{t - 2}$, let us define $Z_{\Cint}$ to be $\{(C_1, C_2) \in Z \mid C_1 \cap C_2 = \Cint\}$. Our observation implies that 
\begin{align} \label{eq:sum-by-t2-clique}
|Z| = \sum_{\Cint \in \binom{V}{t - 2}} |Z_{\Cint}|.
\end{align}
Furthermore, let us define $V_{\Cint} \subseteq V$ to be the set of every vertex such that the union between itself and $\Cint$ belongs to $T$. More formally, let $V_{\Cint} = \cup_{\Cint \subseteq C' \in T} (C' \setminus \Cint)$. It is simple to see that there is a one-to-one correspondence between $E_{G[V_{\Cint}]}$ and $Z_{\Cint}$: for every pair of distinct $C_1, C_2 \in T$ both containing $\Cint$, $(C_1, C_2) \in Z_{\Cint}$ if and only if there is an edge between the vertex $C_1 \setminus \Cint$ and $C_2 \setminus \Cint$. Since $|V_{\Cint}| \leqs |T| \leqs \gamma \cdot k'$, our assumption implies that $|Z_{\Cint}| = |E_{G[V_{\Cint}]}| < k^{1/3} \cdot |V_{\Cint}|$.

Plugging this back to \eqref{eq:sum-by-t2-clique} and rearranging it further, we get
\begin{align*}
|Z| < k^{1/3} \cdot \sum_{{\Cint} \in \binom{V}{t - 2}} |V_{\Cint}|
&= k^{1/3} \cdot \left|\left\{(C', {\Cint}) \in T \times \binom{V}{t - 2} \middle\vert {\Cint} \subseteq C'\right\}\right| \\
&= k^{1/3} \cdot (t - 1) \cdot |T| \\
&\leqs k^{1/3} \cdot (t - 1) \cdot (\gamma \cdot k') \\
&\leqs (\gamma t) k^{t - 2/3} \\
(\text{Since } k \geqs (\lambda\gamma t)^3) &\leqs k^{t - 1/3}/\lambda = \ell/\lambda
\end{align*}
The above lemma together with $|Y| \leqs |Z|$ implies the desired bound in the soundness case.
\end{proof}

\section{From Strong {\sc D{\it k}SH} to {\sc DAL{\it k}S}} \label{sec:dksh-to-dalks}

Finally, we will reduce {\sc D$k$SH} to {\sc DAL$k$S}. The intuition behind the reduction is quite simple. Roughly speaking, the construction is a disjoint union between a clique (with an appropriate size) and the incidence graph\footnote{Recall that an incidence graph of a hypergraph is a bipartite graph between the universe and subsets, whereas there is an edge between each element to every subset containing it.} of the {\sc D$k$SH} instance. The ideal solution we wish is to take the clique together with the elements and subsets in the incidence graph that corresponds to a dense $k$-subhypergraph. The inclusion of the clique ensures that we do not select a subhypergraph of size much larger than $k$, as otherwise it would ``dilute'' the clique's density too much. 

We are now ready to formalize this intuition. Below, we define the gap version of {\sc DAL$k$S} and provide the reduction. The reduction mainly follows the overview above, except that we make copies of vertices in the incidence graph.

\begin{definition}[GapDAL$k$S$(\lambda)$]
For any $\lambda \geqs 1$, the {\sc GapDAL$k$S}$(\lambda)$ problem is: given a graph $G$, $k \in \N$ and $\alpha \in \R^+$, distinguish between (Completeness) $\den_{\geqs k}(G) \geqs \alpha$ and (Soundness) $\den_{\geqs k}(G) < \alpha / \lambda$.
\end{definition}

\begin{lemma} \label{lem:reduction-main}
For any $\varepsilon \in (0, 1]$ and any $t \in \N$, there exists an FPT reduction from $t$-uniform \\ {\sc StrongGapD$k$SH}$(20t^2/\varepsilon, 10^7t^5/\varepsilon^4)$ to {\sc GapDAL$k$S}$(2 - 1/t - \varepsilon)$.
\end{lemma}

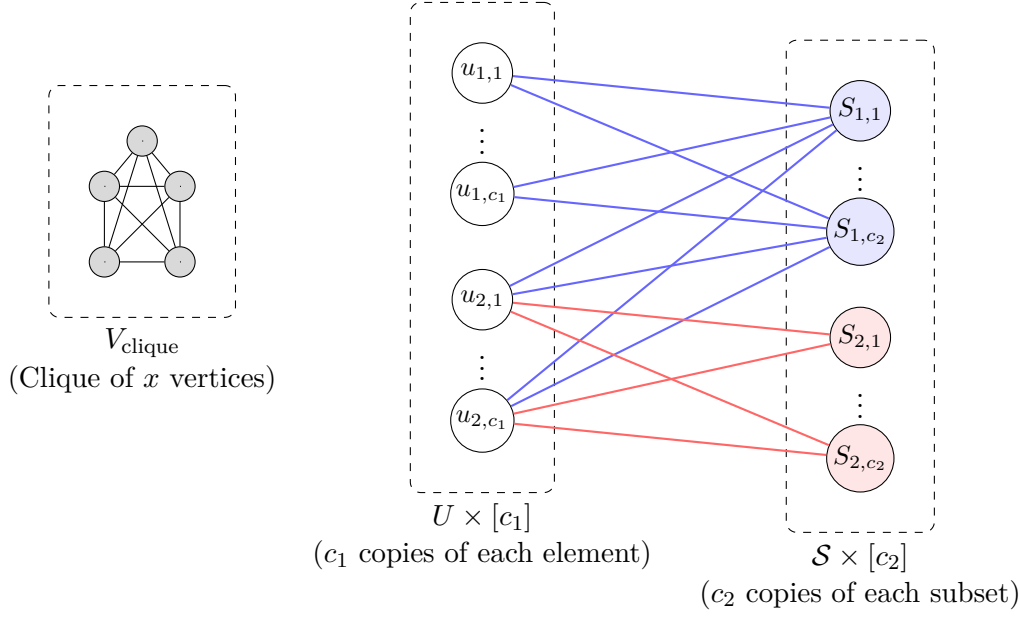
\begin{figure}
\begin{tikzpicture}[
    node distance=1.5cm and 3cm,
    vertex/.style={circle, draw, minimum size=0.8cm, inner sep=1pt, font=\small},
    clique_node/.style={circle, draw, fill=gray!30, minimum size=0.4cm, inner sep=0pt},
    group_box/.style={draw, dashed, rounded corners, inner sep=15pt},
    edge_style/.style={thick, color=black!70}
]

\node[clique_node] (c1) at (0, 1) {};
\node[clique_node] (c2) at (1, 1) {};
\node[clique_node] (c3) at (1, 0) {};
\node[clique_node] (c4) at (0, 0) {};
\node[clique_node] (c5) at (0.5, 1.6) {};

\foreach \i in {1,...,5} {
    \foreach \j in {\i,...,5} {
        \draw (c\i) -- (c\j);
    }
}

\node[fit=(c1)(c2)(c3)(c4)(c5), group_box, label=below:{\shortstack{$V_{\cli}$\\(Clique of $x$ vertices)}}] (clique_box) {};


\node[vertex] (u1_1) at (5, 2.5) {$u_{1,1}$};
\node at (5, 1.7) {$\vdots$};
\node[vertex] (u1_c1) at (5, 0.9) {$u_{1,c_1}$};

\node[vertex] (u2_1) at (5, -0.5) {$u_{2,1}$};
\node at (5, -1.3) {$\vdots$};
\node[vertex] (u2_c1) at (5, -2.1) {$u_{2,c_1}$};

\node[fit=(u1_1)(u2_c1), group_box, label=below:{\shortstack{$U \times [c_1]$\\($c_1$ copies of each element)}}] (U_box) {};

\node[vertex, fill=blue!10] (s1_1) at (10, 2) {$S_{1,1}$};
\node at (10, 1.2) {$\vdots$};
\node[vertex, fill=blue!10] (s1_c2) at (10, 0.4) {$S_{1,c_2}$};

\node[vertex, fill=red!10] (s2_1) at (10, -1) {$S_{2,1}$};
\node at (10, -1.8) {$\vdots$};
\node[vertex, fill=red!10] (s2_c2) at (10, -2.6) {$S_{2,c_2}$};

\node[fit=(s1_1)(s2_c2), group_box, label=below:{\shortstack{$\cS \times [c_2]$\\($c_2$ copies of each subset)}}] (S_box) {};


\draw[edge_style, color=blue!60] (u1_1) -- (s1_1);
\draw[edge_style, color=blue!60] (u1_1) -- (s1_c2);
\draw[edge_style, color=blue!60] (u1_c1) -- (s1_1);
\draw[edge_style, color=blue!60] (u1_c1) -- (s1_c2);

\draw[edge_style, color=blue!60] (u2_1) -- (s1_1);
\draw[edge_style, color=blue!60] (u2_1) -- (s1_c2);
\draw[edge_style, color=blue!60] (u2_c1) -- (s1_1);
\draw[edge_style, color=blue!60] (u2_c1) -- (s1_c2);

\draw[edge_style, color=red!60] (u2_1) -- (s2_1);
\draw[edge_style, color=red!60] (u2_1) -- (s2_c2);
\draw[edge_style, color=red!60] (u2_c1) -- (s2_1);
\draw[edge_style, color=red!60] (u2_c1) -- (s2_c2);
\end{tikzpicture}
\caption{An Illustration of the reduction in Lemma \ref{lem:reduction-main}. The graph $G$ is a disjoint union between an $x$-clique and a repeated version of the incidence graph of $(U, \cS)$ where each element in $U$ is copied $c_1$ times and each subset in $\cS$ is copied $c_2$ times.} \label{fig:reduction-main}
\end{figure}

\begin{proof}
Given an instance $(U, \cS, k, \ell)$ of $t$-uniform {\sc StrongGapD$k$SH}$(20t^2/\varepsilon, 10^7t^5/\varepsilon^4)$, we construct an instance $(G, k', \alpha)$ of {\sc GapDAL$k$S}$(2 - 1/t - \varepsilon)$ as follows. For convenience, let $c_1 := \ell$, $c_2 := \lceil 10^4t^3k/\varepsilon^2\rceil$ and $x := \lceil\sqrt{2(t - 1)\ell c_1 c_2}\rceil$. The vertex set $V_{G}$ of $G$ is a union of two disjoint sets $V_{\cli}$ and $V_H$; there are no edges between the two sets. $V_{\cli}$ consists of $x$ vertices which simply induces a clique. The graph $H := G[V_H]$ is a bipartite graph defined as follows. On one side of the bipartite graph, we have $c_1$ copies of each element of $U$ and, on the other side, we have $c_2$ copies of each subset $S \in \cS$, i.e., $V_H = (U \times [c_1]) \cup (\cS \times [c_2])$. The edges $E_H$ are defined naturally: every copy of $u \in U$ is connected to every copy of each subset containing $u$, i.e., $E_H = \{((u, i), (S, j)) \mid u \in S\}$. Finally, let $k' := x + c_1k + c_2\ell$ and $\alpha = (2t - 1)c_1(1 - 0.1\varepsilon/t)$.

An illustration of the reduction is given in Figure~\ref{fig:reduction-main}. It is obvious that this reduction is FPT and runs in polynomial time. We will next prove its completeness and soundness.

{\bf (Completeness)} Suppose that there exists $T \subseteq U$ of size $k$ such that $|\{S \in \cS \mid S \subseteq T\}| \geqs \ell$. Pick any subset $\cS^*$ of $\{S \in \cS \mid S \subseteq T\}$ of size $\ell$ and let $W = V_{\cli} \cup (T \times [c_1]) \cup (\cS^* \times [c_2])$. Clearly, $|W| = k'$ and, from how the set $W$ is defined, we have
\begin{align*}
\den_{\geqs k'}(G) \geqs \den(G[W]) = \frac{|E_{G[W]}|}{k'} &= \frac{t\ell c_1 c_2 + \binom{x}{2}}{x + c_1k + c_2\ell} \\
&\geqs \frac{(2t - 1)\ell c_1c_2 - x/2}{c_2\ell + x + c_1k} \\
&= (2t - 1)c_1 \left(\frac{1 - \frac{x}{2(2t - 1)\ell c_1c_2}}{1 + \frac{x}{c_2\ell} + \frac{c_1k}{c_2\ell}}\right) \\
&\geqs (2t - 1)c_1\left(1 - \frac{x}{2(2t - 1)\ell c_1c_2}\right)\left(1 - \frac{x}{c_2\ell} - \frac{c_1k}{c_2\ell}\right) \\
&\geqs (2t - 1)c_1\left(1 - \frac{x}{2(2t - 1)\ell c_1c_2} - \frac{x}{c_2\ell} - \frac{c_1k}{c_2\ell}\right) \\
&\geqs (2t - 1)c_1\left(1 - \frac{2x}{c_2\ell} - \frac{c_1k}{c_2\ell}\right) \\
&\geqs (2t - 1)c_1\left(1 - 4\sqrt{2(t - 1)/c_2} - 0.01\varepsilon/t\right) \\
&\geqs (2t - 1)c_1\left(1 - 0.09\varepsilon/t - 0.01\varepsilon/t\right)\\
&\geqs (2t - 1)c_1(1 - 0.1\varepsilon/t) = \alpha.
\end{align*}

{\bf (Soundness)} Suppose that, for every $T \subseteq U$ of size at most $\left(10^7t^5/\varepsilon^4\right) \cdot k$, $|\{S \in \cS \mid |S \cap T| > 1\}| < \ell/\left(20t^2/\varepsilon\right)$. To bound $\den_{\geqs k'}(G)$, we first prove the following intermediate lemma, which gives a bound on the density of subgraphs of $H$ of small sizes.

\begin{lemma} \label{lem:aux-decode}
For every $W' \subseteq V_H$ such that $k'/2 \leqs |W'| \leqs 10tk'/\varepsilon$, $\den(G[W']) < (1 + 0.1\varepsilon)c_1$.
\end{lemma}

\begin{subproof}[Proof of Lemma~\ref{lem:aux-decode}]
For convenience, let us denote $0.1\varepsilon/t$ by $\mu$. Suppose for the sake of contradiction that there exists a set $W' \subseteq V_H$ such that $k'/2 \leqs |W'| \leqs 10tk'/\varepsilon$ and $\den(G[W']) \geqs (1 + t\mu)c_1$. Let $T' := W' \cap (U \times [c_1])$, and let $T \subseteq U$ denote the elements of $U$ such that at least $\mu c_1$ of its copies appear in $T'$, i.e., $T = \{u \in U \mid |(\{u\} \times [c_1]) \cap T'| \geqs \mu c_1\}$. Observe that 
\begin{align*}
|T| \leqs \frac{|T'|}{\mu c_1}
\leqs \frac{10tk'/\varepsilon}{\mu c_1}
\leqs \frac{10t(3c_2\ell)/\varepsilon}{\mu c_1}
\leqs \frac{300000t^4 k \ell/\varepsilon^3}{(0.1\varepsilon/t)\ell} \leqs \left(10^7t^5/\varepsilon^4\right) \cdot k
\end{align*}

Let $\cS_T$ denote $\{S \in \cS \mid |S \cap T| > 1\}$. Next, we claim that $|\cS_T| \geqs \frac{\mu|W'|}{tc_2}$. To see that this is the case, first note that, for each $S \in \cS \setminus \cS_T$, we have $\deg_{G[W']}((S, i)) < (1 + (t - 1)\mu)c_1$ for all $i \in [c_2]$; this is because $(S, i)$ is adjacent to only copies of $u \in S$, and $S \notin \cS_T$ implies that at most one such $u$ has more than $\mu c_1$ copies in $T'$. Equipped with this observation, we can bound $|\cS_T|$ as follows.
\begin{align*}
(1 + t\mu) c_1 \cdot |W'| \leqs \den(G[W']) \cdot |W'| &= |E_{G[W']}| \\
&= \sum_{S \in \cS_T, i \in [c_2] \atop (S, i) \in W'} \deg_{G[W']}((S, i)) + \sum_{S \notin \cS_T, i \in [c_2] \atop (S, i) \in W'} \deg_{G[W']}((S, i)) \\
&= \sum_{S \in \cS_T, i \in [c_2] \atop (S, i) \in W'} tc_1 + \sum_{S \notin \cS_T, i \in [c_2] \atop (S, i) \in W'} (1 + (t - 1)\mu)c_1 \\
&\leqs tc_1c_2|\cS_T| + (1 + (t - 1)\mu)c_1|W'|,
\end{align*}
which immediately implies that $|\cS_T| \geqs \frac{\mu|W'|}{tc_2}$ as previously claimed. As a result, $T$ is a set of size at most $\left(10^7t^5/\varepsilon^4\right) \cdot k$ such that
\begin{align*}
|\cS_T| \geqs \frac{\mu|W'|}{tc_2}
\geqs \frac{(0.1\varepsilon/t)(k'/2)}{tc_2}
\geqs \frac{(0.1\varepsilon/t)(c_2\ell/2)}{tc_2}
\geqs \frac{\ell}{20t^2/\varepsilon},
\end{align*}
which is a contradiction.
\end{subproof}

With the above lemma ready, it is now quite easy to bound $\den_{\geqs k'}(G)$. Before we do so, let us first derive the following useful lower bound for $\frac{\alpha}{2 - 1/t - \varepsilon}$.
\begin{align*}
\frac{\alpha}{2 - 1/t - \varepsilon} = \frac{tc_1(1 - 0.1\varepsilon/t)}{1 - \frac{t \varepsilon}{2t - 1}} \geq \frac{tc_1(1 - 0.1\varepsilon/t)}{1 - 0.5\varepsilon/t} \geq tc_1(1 + 0.4\varepsilon/t).
\end{align*}
where the second inequality comes from the fact that $\frac{1 - \nu_1}{1 - \nu_2} \geqs 1 + (\nu_2 - \nu_1)$ for all $0 \leqs \nu_1 \leqs \nu_2 < 1$.

Finally, we will now bound $\den_{\geqs k'}(G)$. Let $W$ be any subset of $V_{G}$ of size at least $k'$. Let $W' = W \setminus V_{\cli}$. To bound $\den(G[W])$, consider the following two cases.
\begin{enumerate}
\item $|W| \leqs 10tk'/\varepsilon$. Observe that $k' \geqs 2x$, meaning that $|W'| \geqs k'/2$. Thus, we have
\begin{align*}
\den(G[W]) = \frac{|E_{G[W]}|}{|W|} 
\leqs \frac{\binom{x}{2}}{k'} + \den(G[W'])
&\overset{(\text{Lemma~\ref{lem:aux-decode}})}{<} (t - 1)c_1 + c_1(1 + 0.1\varepsilon) \\ &= tc_1(1 + 0.1\varepsilon/t) \\ &< \frac{\alpha}{2 - 1/t - \varepsilon}.
\end{align*}
\item $|W| > 10tk'/\varepsilon$. Note that, since $H$ is a bipartite graph whose vertices on one side have degree $tc_1$, the density of any subgraph of $H$ is less than $tc_1$. Hence, we have
\begin{align*}
\den(G[W]) = \frac{|E_{G[W]}|}{|W|} 
< \frac{\binom{x}{2}}{10tk'/\varepsilon} + \den(H[W'])
&< 0.1\varepsilon c_1 + tc_1 \\
&\leqs tc_1(1 + 0.1\varepsilon/t) \\
&< \frac{\alpha}{2 - 1/t - \varepsilon}.
\end{align*}
\end{enumerate}
Thus, in both cases, $\den(G[W]) < \frac{\alpha}{2 - 1/t - \varepsilon}$, which implies that $\den_{\geqs k'}(G) < \frac{\alpha}{2 - 1/t - \varepsilon}$ as desired.
\end{proof}

\section{Putting Things Together: Proof of Theorems~\ref{thm:dks} and~\ref{thm:polydks}}

With all the ingredients ready, we can now prove our main theorems (Theorems~\ref{thm:dks} and~\ref{thm:polydks}).

\begin{proof}[Proof of Theorem~\ref{thm:dks}]
For any $\eps > 0$, let $\lambda = 20^{16} / \eps^9$.
From Proposition~\ref{prop:size-relax}, Lemma~\ref{lem:dks-to-dksh}, and Lemma~\ref{lem:reduction-main}, there are FPT polynomial time reductions for the following sequence of problems.
\begin{itemize}
\item {\sc GapDkS}$(\lambda)$
\item {\sc GapDkS}$(80/\varepsilon, 20^7/\varepsilon^4)$
\item 2-uniform {\sc StrongGapD$k$SH}$(80/\varepsilon, 20^7/\varepsilon^4)$
\item  {\sc GapDAL$k$S}$(3/2 - \varepsilon)$
\end{itemize}
Thus, if there is an FPT (resp. polynomial) time $(3/2 - \eps)$-approximation algorithm for {\sc DAL$k$S}, we can solve {\sc GapDkS}$(\lambda)$ in FPT (resp. polynomial) time.
\end{proof}

The proof of Theorem~\ref{thm:polydks} is nearly the same as above, except that we now use Proposition~\ref{prop:size-relax-poly} and Lemma~\ref{lem:polydks-to-dksh} instead of Proposition~\ref{prop:size-relax} and Lemma~\ref{lem:dks-to-dksh}, respectively.

\begin{proof}[Proof of Theorem~\ref{thm:polydks}]
For any $\eps > 0$, let $t = \lceil 2/\eps\rceil$ and $\gamma = 20^7t^5/\varepsilon^4$.
From Proposition~\ref{prop:size-relax-poly}, Lemma~\ref{lem:polydks-to-dksh}, and Lemma~\ref{lem:reduction-main}, there are FPT polynomial time reductions for the following sequence of problems.
\begin{itemize}
\item {\sc PolyGapDkS}$(\frac{1}{6t}, t)$
\item {\sc PolyGapDkS}$(1/3, t, \gamma)$
\item $t$-uniform {\sc StrongGapD$k$SH}$(40t^2/\varepsilon, \gamma)$
\item  {\sc GapDAL$k$S}$(2 - \eps)$
\end{itemize}
Thus, if there is an FPT (resp. polynomial) time $(2 - \eps)$-approximation algorithm for {\sc DAL$k$S}, we can solve {\sc PolyGapDkS}$(\frac{1}{6t}, t)$ in FPT (resp. polynomial) time.
\end{proof}

\section{\W[1]-Hardness of Exact {\sc DAL{\it k}S}}

Finally, we show that the exact version of {\sc DAL$k$S} is W[1]-hard. We reduce from the $k$-Clique problem, which is well known to be \W[1]-hard~\cite{DowneyF95}. The reduction here is essentially the same as the one used in Section~\ref{sec:dksh-to-dalks}, except that we only use one copy of the incidence graph. However, we need to be more careful in our argument below as there is no gap in the starting problem.

\begin{proof}[Proof of Theorem~\ref{thm:w1}]
Given an instance $(G, k)$ of $k$-Clique, we construct an instance $(G', k', \alpha)$ of {\sc DAL$k$S} as follows. The vertex set $V_{G'}$ of $G'$ is a union of two disjoint sets $V_{\cli}$ and $V_H$ with no edges between them. $V_{\cli}$ consists of $x := k$ vertices which induces a clique. The graph $H := G'[V_H]$ is the incidence graph of $G$, i.e., a bipartite graph where one set of vertices is the set of vertices of $G$, the other is the set of edges of $G$, and there is an edge between $u \in V_G$ and $e \in E_G$ iff $u \in e$. Finally, we let $k' = x + k + \binom{k}{2}$ and $\alpha = \left(\binom{x}{2} + 2\binom{k}{2}\right)/k'$. We assume w.l.o.g. that $k \geqs 10$. Note that this implies that $\alpha > 2$.

{\bf (Completeness)} Suppose that there exists $T \subseteq V_G$ that induces a $k$-clique in $G$. Let $S' = V_{\cli} \cup T \cup E_{G[T]}$. Clearly, $|S'| = k'$ and, from how the set $S'$ is defined, we have
\begin{align*}
\den(G'[S']) = \frac{|E_{G'[S']}|}{k'} = \frac{|E_{G'[V_{\cli}]}| + |E_{G'[T \cup E_{G[T]}]}|}{k'} = \frac{\binom{x}{2} + 2\binom{k}{2}}{k'} = \alpha.
\end{align*}

{\bf (Soundness)} We will show the contrapositive. Suppose that $\den_{\geqs k'}(G') \geqs \alpha$; let $S' \subseteq V_{G'}$ be any subset of vertices such that $\den(G'[S']) = \den_{\geqs k'}(G')$.

First, we claim that $V_{\cli} \subseteq S'$. Suppose for the sake of contradiction that $|S' \cap V_{\cli}| = y < x$. Observe that $\den(G'[S'])$ is at most $$\max\{\den(G'[S' \cap V_{\cli}]), \den(G'[S' \cap V_H])\} = \max\{(y-1)/2, 2\} < \frac{x + y - 1}{2}.$$
Consider the set $S'' = S' \cup (V_{\cli} \setminus S')$. The number of vertices in $S'' \setminus S'$ is $x - y$ and the number of edges in $E_{G'[S'']} \setminus E_{G'[S']}$ is $\binom{x}{2} - \binom{y}{2}$. Hence, the ratio between the number of added edges and the number of added vertices is
\begin{align*}
\frac{\binom{x}{2} - \binom{y}{2}}{x - y} = \frac{x + y - 1}{2} > \den(G'[S']).
\end{align*}
This means that $\den(G'[S'']) > \den(G'[S'])$, which is a contradiction to the fact that $S'$ is the subset of vertices that maximizes the density of $G'[S']$ among all subsets of size at least $k'$.

Secondly, we claim that $|S'| = k'$. Suppose for the sake of contradiction that $|S'| > k'$. Since every vertex on one side (the edge side) of the bipartite graph $H$ has degree two, there must be a vertex $u' \in S' \cap V_H$ that has degree at most two (in $G'[S']$). Since $\alpha > 2$, the graph $G'[S' - \{u'\}]$ has a higher density than $G'[S']$, which once again is a contradiction to the definition of $S'$.

Now, consider the set $T' = S' \setminus V_{\cli}$. From the above two properties, we have $|T'| = k + \binom{k}{2}$ and, from $\den(G'[S']) \geqs \alpha$, we have $|E_{H[T']}| \geqs 2\binom{k}{2}$. Let $T'_{\text{vertex}} := T' \cap V_G$ and $T'_{\text{edge}} := T' \cap E_G$. We claim that $|T'_{\text{vertex}}| = k$ and that $T'_{\text{vertex}}$ induces a clique in $G$. To prove this, since $H$ is bipartite and each $e \in T'_{\text{edge}}$ has degree at most two, we can rearrange $|E_{H[T']}|$ as follows.
\begin{align*}
|E_{H[T']}| = \sum_{e \in T'_{\text{edge}}} \deg_{H[T']}(e)
&\leqs |T'_{\text{edge}}| + |\{e \in T'_{\text{edge}} \mid \deg_{H[T']}(e) = 2\}| \\
&\leqs |T'_{\text{edge}}| + \min\{|E_{G[T'_{\text{vertex}}]}|, |T'_{\text{edge}}|\}.
\end{align*}
Since $|E_{H[T']}| \geqs 2\binom{k}{2}$, we have $|T'_{\text{edge}}| \geqs \binom{k}{2}$ and hence $|T'_{\text{vertex}}| \leqs k$. Let $a = |T'_{\text{vertex}}| \leqs k$; the above inequality can be rearranged further as follows.
\begin{align*}
|E_{H[T']}| \leqs |T'_{\text{edge}}| + |E_{G[T'_{\text{vertex}}]}| \leqs \left(k + \binom{k}{2} - a\right) + \binom{a}{2} = 2\binom{k}{2} - (k - a)\left(\frac{a + k - 3}{2}\right) \leqs 2\binom{k}{2}.
\end{align*}
Since $|E_{H[T']}| \geqs 2\binom{k}{2}$, all inequalities must be equalities. Specifically, we must have $k = a$ and $|E_{G[T'_{\text{vertex}}]}| = \binom{k}{2}$. That is, $T'_{\text{vertex}}$ is a set of $k$ vertices that induce a clique in $G$ as desired.
\end{proof}

\section{Concluding Remarks}

In this work, we provide a reduction from (variants of) {\sc D$k$S} to {\sc DAL$k$S}, which establishes hardness of approximation for the latter under several assumptions. An interesting open question is whether we can strengthen our second reduction Theorem~\ref{thm:polydks}, which currently requires a rather strong assumption that {\sc PolyGapDkS} is hard. A concrete direction here is to attempt to relax the assumption to, e.g., that {\sc GapDkS} is hard.

\section*{Acknowledgments}

This work was partially supported by the National Natural Science Foundation of China (NSFC) under Grant No. 62572310.

\bibliographystyle{alpha}
\bibliography{main}

\appendix

\end{document}